\documentclass[12pt, oneside]{article} 
\usepackage{graphicx}
\usepackage[margin=0.5in]{geometry}
\usepackage{amssymb}
\usepackage{amsmath}
\usepackage{amsthm}
\usepackage[colorlinks]{hyperref}
\newcommand{\BEQ}{::=}
\newcommand{\BOR}{\mid}
\newcommand{\type}{\!:\!}
\newcommand{\R}{\mathbb{R}}
\newcommand{\apdiff}[3]{\left .\frac{\partial #1}{\partial #2} \right\vert_{#2 = #3}}
\newcommand{\pdiff}[2]{\frac{\partial #1}{\partial #2}}
\newcommand{\pdiffop}[1]{\frac{\partial}{\partial #1}}
\newcommand{\ov}[1]{\overline{#1}}

\newcommand{\FV}{\mathrm{FV}}
\newcommand{\FnV}{\mathrm{FnV}}

\newcommand{\op}{\mathrm{op}}
\newcommand{\den}[1]{\mathcal{S}[\![ #1 ]\!]}

\newcommand{\CF}{\mathrm{CF}}

\newcommand{\lskip}{$\mbox{ }$\vspace{-10pt}\\}
\newcommand{\Imm}{\mathrm{Im}}
\newcommand{\dapprox}{\,\dot{\approx}\,}

\newcommand{\rhoN}{\rho_{\mathbb{N}}}
\newcommand{\varphiN}{\varphi_{\mathbb{N}}}

\newcommand{\acomment}[1]{}
\newcommand{\myomit}[1]{}
\newcommand{\realalpha}{\alpha}
\newcommand{\realbeta}{\beta}

\newcommand{\funalpha}{f }
\newcommand{\funbeta}{g }
\newcommand{\fungamma}{h }
\newcommand{\eqdef}{=_{\tiny \mathrm{def}}}
\newcommand{\rdash}{\vdash_{{\tiny \mathrm{RTC}}}}
\newcommand{\B}{\mathbb{B}}

\newtheorem{lemma}{Lemma}
\newtheorem{theorem}{Theorem}
\newtheorem{corollary}{Corollary}

\title{A complete equational axiomatisation of partial differentiation\hspace{-5pt}
\footnote{This is a version with minor corrections of a paper given at MFPS 2020.}}
\author{Gordon D.\ Plotkin}

  \date{Google Research\\Mountain View, United States}
\begin{document}
  \maketitle

  
\begin{abstract} 
  We formalise the well-known  rules of partial differentiation in a version of equational logic with function variables and binding constructs. We prove the resulting theory is complete with respect to polynomial interpretations. The proof makes use of Severi's interpolation theorem that all multivariate Hermite problems are solvable. We also present a number of  related results, such as decidability and equational completeness.

\end{abstract}

\section{Introduction}



There has been recent increasing interest in categorical axiomatisations of differential structure. For example, for forward differentiation, see~\cite{BCS09}; for reverse differentiation see~\cite{CCG20}; and for tangent structures, see~\cite{CC14}. A natural question is whether the axioms are well chosen. The authors generally show they hold in natural structures. For example, in the case of cartesian differential categories they hold for the category of finite powers of the reals and smooth functions.  So the axioms are, in that sense, correct. But one can additionally ask, if, or in what sense, they are complete, that is whether there are missing axioms.

Here we interest ourselves in an allied basic logical question which we hope will help with the categorical one: are the standard rules for  manipulating partial derivatives  complete?
%
%
The rules are indeed well known: derivatives of products and sums are given by simple formulas involving their immediate subexpressions and their partial derivatives;  derivatives of real constants and variables are 0 or 1; the chain rule takes care of function applications; and partial derivatives with respect to different variables commute. They are surely complete.



To prove this,  we proceed by first formalising the standard rules of partial differentiation using a suitable  kind of equational logic. Our axiom system consists of the ring axioms, addition and multiplication tables for the reals, and four axioms for partial differentiation.
As partial differentiation involves variable binding, we are outside the scope of standard equational logic. So we instead use an expanded version which allows binding constructs and makes use of function variables.  An equational  logic of this kind, called \emph{second-order equational logic}, has been presented by Fiore, Hur, and Mahmoud~\cite{FH10,FM10}. We employ a minor, although entirely equivalent, variant of their logic. 

The resulting theory of partial differentiation is naturally interpreted using smooth functions, but, for completeness, it turns out that it suffices to use only polynomial functions. Indeed, the theory is complete even if the interpretation of function variables is restricted  to natural number polynomial functions, and the interpretation of  variables to natural numbers. This remains the case if the constants are restricted to the rationals or to $0$ and $1$ (equivalently the integers).

To establish completeness we employ an interpolation theorem, a well-known theorem of Severi~\cite{Sev21} on the solvability of Hermite interpolation problems. Hermite interpolation generalises Lagrange interpolation. In Lagrange interpolation, one seeks a real polynomial taking prescribed values at prescribed points, and this has an evident generalization to multivariate Lagrange interpolation. In multivariate  Hermite interpolation,  one additionally prescribes values of (possibly higher-order) partial derivatives. 

Theorem~\ref{complete}  establishes completeness relative to natural number polynomial interpretations of function variables and natural number interpretations of variables. (The polynomials are adaptations of real polynomials solving suitable Hermite interpolation problems.) We also give a number of other results. 
Theorem~\ref{analysis} characterises the theorems of our equational theory in terms of  an equivalence relation between canonical forms. 
Theorem~\ref{RTC} shows that the equivalence relation holds if, and only if, it can be established using just one axiom for partial differentiation, that the order of partial differentiation by two variables does not matter (with a slightly more refined notion of canonical form this axiom can be eliminated, when no axioms for partial differentiation are needed).
Theorem~\ref{HP} shows that the theory is not only complete with respect to standard interpretations over the reals, it is also equationally complete (also known as Hilbert-Post complete), that is, it has no equationally consistent proper extensions, see~\cite{Tay79}. This remains the case if the constants are restricted to the rationals but not if they are restricted to the integers. 
Theorem~\ref{decidable} shows that the subtheory in which the constants are restricted to the rationals is decidable, and that, in case an equation does not hold, a counterexample to it can be found.  An important remaining question is the complexity of the decision problem.

\section{Axiomatisation of partial differentiation}

We assume disjoint sets of  variables, ranged over by $x,y,z,\ldots$, and of function variables, ranged over by $\funalpha,\funbeta,\fungamma,\ldots$; each function variable has a given arity $n \geq 0$, written $\funalpha\type n$. The set of variables is  assumed denumerable, as is the set of function variables of each  arity. Expressions have the following forms:
\[e \BEQ r \quad (r \in \R) \BOR x \BOR e_0 + e_1 \BOR e_0e_1 \BOR \funalpha(e_0,\ldots,e_{n-1}) \quad (\funalpha\type n) \BOR \mathrm{PDiff}(x .e_0, e_1)\]
The expression $ \mathrm{PDiff}(x.e_0, e_1)$ is read as the partial derivative of $e_0$ with respect to $x$, evaluated at $e_1$. In this expression, the variable $x$ has binding power over $e_0$, but not over $e_1$. Below, we use the more familiar, and more natural, expression  $ \apdiff{e_0}{x}{e_1}$ (while more natural, this expression obscures  the role of the variable $x$). We write $\pdiffop{x}e$ for $\apdiff{e}{x}{x}$, and note that the last occurrence of $x$ is free; we write $\pdiff{^ne}{x_{n-1}\ldots \partial x_0}$ for $\pdiffop{x_{n-1}}\ldots \pdiffop{x_0}e$.

Free variables and $\realalpha$-equivalence, are defined as usual, and,  as is also usual, we identify $\realalpha$-equivalent expressions. 
We write $\FV(e)$ for the set of free variables of an expression $e$ and $\FnV(e)$ for the set of its function variables.  Every expression $e$ has a size $|e|$, defined in an evident way.

The simultaneous substitution \[e[e'_0/x_0,\ldots,e'_{n-1}/x_{n-1}]\] of expressions for variables is defined as may be expected, with the clause for $ \mathrm{PDiff}(x.e_0, e_1)$, with its bound variable, being:
\[\begin{array}{lcl}\hspace{-13pt}
 \mathrm{PDiff}(x.e_0, e_1)[e'_0/x_0,\ldots,e'_{n-1}/x_{n-1}] & = &  \mathrm{PDiff}(x.e_0[e'_0/x_0,\ldots,e'_{n-1}/x_{n-1}], 
 e_1[e'_0/x_0,\ldots,e'_{n-1}/x_{n-1}])\\
\end{array}\]
where $x \notin  \FV(e'_0) \cup \ldots \cup  \FV(e'_{n-1})$. 
There is also a notion of substitution of \emph{abstracts} $(x_0,\ldots,x_{n-1}).e'$ for function variables.
For   $\funalpha\type n $, we define 
\[e[(x_0,\ldots,x_{n-1}).e'/\funalpha]\]
 by structural induction on $e$ as follows:\\
\[\begin{array}{lcl} 
r[(x_0,\ldots,x_{n-1}).e'/\funalpha] & = & r \\\\

x[(x_0,\ldots,x_{n-1}).e'/\funalpha] & = & x\\\\

(e_0 + e_1)[(x_0,\ldots,x_{n-1}).e'/\funalpha] & = & 
            e_0[(x_0,\ldots,x_{n-1}).e'/\funalpha] + e_1[(x_0,\ldots,x_{n-1}).e'/\funalpha]\\\\
            
(e_0 e_1)[(x_0,\ldots,x_{n-1}).e'/\funalpha] & = & 
             e_0[(x_0,\ldots,x_{n-1}).e'/\funalpha] e_1[(x_0,\ldots,x_{n-1}).e'/\funalpha]\\\\
             
\funbeta(e_0,\ldots,e_{n-1})[(x_0,\ldots,x_{n-1}).e'/\funalpha] & = &
                \left \{\begin{array}{ll}
                                     e'[e_0[(x_0,\ldots,x_{n-1}).e'/\funalpha]/x_0,\ldots,\\ \hspace{50pt} e_{n-1}[(x_0,\ldots,x_{n-1}).e'/\funalpha]/x_{n-1}] \\  \hspace{200pt}(\funbeta = \funalpha)\\\\
                                     \funbeta(e_0[(x_0,\ldots,x_{n-1}).e'/\funalpha],\ldots, \\ \hspace{50pt} 
                                            e_{n-1}[(x_0,\ldots,x_{n-1}).e'/\funalpha]) \\  
                                            \hspace{200pt} (\funbeta \neq \funalpha) \\
                            \end{array}
                 \right .\\\\
                 
\mathrm{PDiff}(x.\, e_0, e_1)[(x_0,\ldots,x_{n-1}).e'/\funalpha] & = &    \mathrm{PDiff}(x.\, e_0[(x_0,\ldots,x_{n-1}).e'/\funalpha], 
 e_1[(x_0,\ldots,x_{n-1}).e'/\funalpha])\\
&& \hfill (x\notin \FV(e')\backslash\{x_0,\ldots,x_{n-1}\})
\end{array}\]
%
%
There is a more general, similarly defined, simultaneous substitution of several abstracts: 
\[e[(x_{11},\ldots,x_{1n_1}).e_1'/\funalpha_1, \ldots, (x_{k1},\ldots,x_{kn_k}).e_k'/\funalpha_k]\]
where $\funalpha_i\type n_i$, for $i = 1,k$.

The axioms of the theory of partial differentiation consist of those for commutative rings, the addition and multiplication tables for the real constants, and  four axioms for partial differentiation. The four axioms are the usual rules for the differentiation of addition and multiplication, the chain rule for binary functions, and the commutativity of partial differentiation with respect to different  variables.

\lskip

\lskip

\[\pdiff{\,x + y}{x} = 1 \qquad \qquad \pdiff{yx}{x} = y\]

\lskip

\[ \pdiff{\,\funalpha(\funbeta_0(x),\funbeta_{1}(x))}{x} \;\;  = \;\;
      \apdiff{\funalpha(x_0,\funbeta_{1}(x))}{x_0}{\funbeta_0(x)}\pdiff{\funbeta_0(x)}{x} \; + \; \apdiff{\funalpha(\funbeta_0(x),x_{1})}
   {x_{1}}{\funbeta_{1}(x)}\pdiff{\funbeta_{1}(x)}{x}
  \]
  
  \lskip 
  
\[\pdiffop{y}\pdiffop{x}\funalpha(x,y) \;\;  = \;\; \pdiffop{x}\pdiffop{y}\funalpha(x,y)\]

\lskip

%
%

%

%
%

The theorems of the theory of partial differentiation are obtained by closing the axioms under the evident  rules for equality, congruence, and both kinds of substitution, following the usual pattern for equational systems. In particular, the congruence rule for partial differentiation is as follows: 
\[\frac{e_0 = e_1 \qquad e'_0 = e'_1}{ \mathrm{PDiff}(x.\, e_0, e'_0) =  \mathrm{PDiff}(x.\, e_1, e'_1)}\]
If an equation $e_0 = e_1$ holds in this theory  we write
\[\vdash e_0 = e_1\]

We next establish a number of expected consequences of the axioms, including a general form of the chain rule. We employ an evident abbreviation $\sum_{i=m}^{n}e_i$ for finite sum expressions.

\begin{lemma} \label{facts} The following equations are provable in the theory of partial differentiation:
\begin{enumerate}
\item \[ \pdiff{x}{x} = 1\]
\item \[\pdiff{\,\funalpha(x) + \funbeta(x)}{x}\;\;  = \;\; \pdiff{\funalpha(x)}{x} + \pdiff{\funbeta(x)}{x} \]
\item \[ \pdiff{\,\funalpha(x)\funbeta(x)}{x} \;\;  = \;\; \funbeta(x)\pdiff{\funalpha(x)}{x} + \funalpha(x)\pdiff{\funbeta(x)}{x}\]
\item 
\[ \pdiff{\,\funalpha(\funbeta_0(x),\ldots,\funbeta_{n-1}(x))}{x}  \;\;  = \;\; 
\sum_{i = 0}^{n-1}\apdiff{\funalpha(\funbeta_0(x), \ldots, \funbeta_{i-1}(x),x_i, \funbeta_{i+1}(x),\ldots,\funbeta_{n-1}(x))}{x_i}{\funbeta_i(x)}\pdiff{\funbeta_i(x)}{x}\]

%

\item
 \[\pdiff{e}{x} = 0  \quad (x \notin \FV(e))\]
 \end{enumerate}
\end{lemma}
\begin{proof}
\begin{enumerate}
\item  This follows from the axiom for addition, substituting $0$ for $y$.
\item  This follows from the axiom for addition and the binary chain rule.
\item   This follows from the axiom for multiplication and the binary chain rule.
\item  First, note that, substituting $0$ for $y$ in the axiom for multiplication, we obtain $\pdiff{0}{x} = 0$. 
Next, for the case $n = 0$ of the chain rule, substitute $(x_0,x_1).\,\fungamma()$ for $\funalpha$ and $(z).\,0$ for both $\funbeta_0$ and $\funbeta_1$ in the binary chain rule to get the following provable equations:
 \[\begin{array}{lcl} \pdiff{\fungamma()}{x} & \;\;  = \;\; & 
      \apdiff{\fungamma()}{x_0}{0}\pdiff{0}{x} \; + \; \apdiff{\fungamma()}{x_{1}}{0}\pdiff{0}{x}\\
     & \;\;  = \;\; &        \apdiff{\fungamma()}{x_0}{0}\mbox{{\small 0}} \; + \; \apdiff{\fungamma()}{x_{1}}{0}\mbox{{\small 0}}\\
     & \;\;  = \;\; &  \mbox{{\small 0}}
  \end{array}\]

Next,  for the unary chain rule, substitute $(x_0,x_1).\,\fungamma(x_0)$ for $\funalpha$ and $(z).\,0$ for $\funbeta_1$ in the binary chain rule to get the following provable equations:
 \[\begin{array}{lcl} \pdiff{\fungamma(\funbeta_0(x_0))}{x} & \;\;  = \;\; & 
      \apdiff{\fungamma(x_0)}{x_0}{\funbeta_0(x)}\pdiff{\funbeta_0(x)}{x} \; + \; \apdiff{\fungamma(x_0)}{x_{1}}{0}\pdiff{0}{x}\\
     & \;\;  = \;\; &        \apdiff{\fungamma(x_0)}{x_0}{\funbeta_0(x)}  \pdiff{\funbeta_0(x)}{x} 
       \end{array}\]

Finally, for $n \geq 2$, we proceed by induction. The axiom provides the base case, so assuming the induction hypothesis for $n$, we prove it for $n + 1$. Substituting $(z).\,z$ for $\funbeta_0$ and $\funbeta_1$ in the binary chain rule, and using $\pdiff{x}{x} = 1$, we obtain the following equation:
\[ \pdiff{\,\funalpha(x,x)}{x} \;\;  = \;\;
      \apdiff{\funalpha(x'_0,x)}{x'_0}{x} \; + \; \apdiff{\funalpha(x,x'_{1})}{x'_{1}}{x}
  \]
Then, substituting $(x,y).\,\fungamma(\funbeta_0(x),\funbeta_1(y), \ldots, \funbeta_n(y))$ for $\funalpha$, we obtain:
\[ \pdiff{\,\fungamma(\funbeta_0(x),\funbeta_1(x), \ldots, \funbeta_n(x))}{x} \;  = \;
      \apdiff{\fungamma(\funbeta_0(x'_0),\funbeta_1(x), \ldots, \funbeta_n(x))}{x'_0}{x} \, + \,
       \apdiff{\fungamma(\funbeta_0(x),\funbeta_1(x'_1), \ldots, \funbeta_n(x'_1))}{x'_{1}}{x}
  \]
For the first of these two summands we have:
\[ \begin{array}{lcl}\apdiff{\fungamma(\funbeta_0(x'_0),\funbeta_1(x), \ldots, \funbeta_n(x))}{x'_0}{x} 
& = &
\pdiff{\fungamma(\funbeta_0(x'_0),\funbeta_1(x), \ldots, \funbeta_n(x))}{x'_0} [x/x'_0]\\
& = & \left ( \apdiff{\fungamma(x_0,\funbeta_1(x), \ldots, \funbeta_n(x))}{x_0}{\funbeta_0(x'_0)}\pdiff{\funbeta_0(x'_0)}{x'_0}\right ) [x/x'_0]\\
& = & \apdiff{\fungamma(x_0,\funbeta_1(x), \ldots, \funbeta_n(x))}{x_0}{\funbeta_0(x)}\pdiff{\funbeta_0(x)}{x}
 \end{array}\]
with the second equality being an application of the unary chain rule and the other two equalities being syntactic identities (the second up to $\realalpha$-equivalence).

For the second of these, making use of the induction hypothesis (the $n$-ary chain rule) we have:
\[\begin{array}{lcl} \apdiff{\fungamma(\funbeta_0(x),\funbeta_1(x'_1), \ldots, \funbeta_n(x'_1))}{x'_1}{x} 
&\hspace{-10pt} = &
 \pdiff{\fungamma(\funbeta_0(x),\funbeta_1(x'_1), \ldots, \funbeta_n(x'_1))}{x'_1}[x/x'_1]\\
 & \hspace{-10pt} = & \left  ( \sum_{i = 1}^{n}\apdiff{\fungamma(\funbeta_0(x),\funbeta_1(x'_1), \ldots, \funbeta_{i-1}(x'_1),x_i, \funbeta_{i+1}(x'_1),\ldots,\funbeta_{n}(x'_1))}{x_i}{\funbeta_i(x'_1)}\pdiff{\funbeta_i(x'_1)}{x'_1}\right )[x/x'_1]\\
& \hspace{-10pt} = &  \sum_{i = 1}^{n}\apdiff{\fungamma(\funbeta_0(x),\funbeta_1(x), \ldots, \funbeta_{i-1}(x),x_i, \funbeta_{i+1}(x),\ldots,\funbeta_{n}(x))}{x_i}{\funbeta_i(x)}\pdiff{\funbeta_i(x)}{x}
\end{array}\]
Combining the results, we reach the desired conclusion.

  \item For $n = 0$, the chain rule  is $\pdiff{\funalpha()}{x} = 0$, where $\funalpha \type 0$.
 Substituting $(\;).\,e$ for $\funalpha$ we obtain $\pdiff{e}{x} = 0$ as $x \notin \FV(e)$.
\end{enumerate}
\end{proof}

It will prove useful to know the interaction between differentiation and substitution:
\begin{lemma} \label{subdiff} The following are equivalent:
\begin{enumerate}
\item The general chain rule.
\item The following substitution principle, that for all expressions $e$ and $e_0,\ldots,e_{n-1}$ we have:
%
\[\vdash \pdiff{e[e_0/x_0,\ldots,e_{n-1}/x_{n-1}] }{x} \;\; =  \;\;
\sum_{i = 0}^{n-1} \pdiff{e}{x_i}[e_0/x_0,\ldots,e_{n-1}/x_{n-1}] \pdiff{e_i}{x}\]
where $x_0,\ldots, x_{n-1}$ are distinct variables such that $x \notin \FV(e)\backslash \{x_0,\ldots, x_{n-1}\}$.
\end{enumerate}
\end{lemma}
\begin{proof}
In one direction assume the general chain rule.
Choose $\funalpha \type n$ not in any of the $\FnV(e_i)$. Using the general chain rule we have:
\[\vdash
\pdiff{\funalpha(e_0,\ldots,e_{n-1}) }{x}  \;\; = \;\; \sum_{i = 0}^{n-1}\apdiff{\funalpha(e_0,\ldots,e_{i-1},x_i, e_{i +1},\ldots,e_{n-1})}{x_i}{e_i}\pdiff{e_i}{x}
\]
%
Then, substituting $(x_0,\ldots,x_{n-1}).\,e $ for $\funalpha$ on the left we have the following identities (recall that we identify $\realalpha$-equivalent expressions):
\[\begin{array}{lcll}
\pdiff{\funalpha(e_0,\ldots,e_{n-1}) }{x}[(x_0,\ldots,x_{n-1}).\,e/\funalpha] & = & 
  \pdiff{\funalpha(e_0,\ldots,e_{n-1}) [(x_0,\ldots,x_{n-1}).\,e/\funalpha] }{x} & (\mbox{\small as  $x\notin \FV(e)\backslash\{x_0,\ldots,x_{n-1}\}$})\\
  & = &    \pdiff{e[e_0/x_0,\ldots,e_{n-1}/x_{n-1}] }{x} & 
  (\mbox{\small as  $\funalpha \notin \FnV(e_0, \ldots, e_{n-1}$}))
\end{array}
\]
%

Next, substituting on the right, we have the identities:
\[\begin{array}{lll}
\left ( \sum_{i = 0}^{n-1}\apdiff{\funalpha(e_0,\ldots,e_{i-1},x_i, e_{i +1},
\ldots,e_{n-1})}{x_i}{e_i}\pdiff{e_i}{x}\right )[(x_0,\ldots,x_{n-1}).\,e/\funalpha] \\
\hspace{100pt}  \;\; = \;\; \sum_{i = 0}^{n-1}\left (\apdiff{\funalpha(e_0,\ldots,e_{i-1},x_i, e_{i +1},
\ldots,e_{n-1})}{x_i}{e_i}[(x_0,\ldots,x_{n-1}).\,e/\funalpha]
\right )\pdiff{e_i}{x}  
\\ \hspace{300pt}(\mbox{\small as  $\funalpha \notin \FnV(e_i)$, for $i = 0,\ldots, n-1$}) \\
 \hspace{100pt} \;\; = \;\; \sum_{i = 0}^{n-1}\left ( \apdiff{\funalpha(e_0,\ldots,e_{i-1},x_i, e_{i +1},
\ldots,e_{n-1})[(x_0,\ldots,x_{n-1}).\,e/\funalpha]}{x_i}{e_i} \right )\pdiff{e_i}{x} 
\\ \hspace{300pt}(\mbox{\small as $x_i\notin \FV(e)\backslash \{x_0,\ldots, x_{n-1}\}$})\\
 \hspace{100pt}  \;\; = \;\; \sum_{i = 0}^{n-1} \apdiff{e[e_0/x_0,\ldots,e_{i-1}/x_{i-1},x_i/x_i, e_{i 
 +1}/x_{i+1},\ldots,e_{n-1}/x_{n-1}]}{x_i}{e_i} \pdiff{e_i}{x}  
\\ \hspace{300pt}(\mbox{\small as $\funalpha \notin \FnV(e_0,\ldots,e_{i-1}, e_{i +1},\ldots,e_{n-1})$})
\\
 \hspace{100pt}  \;\; = \;\; \sum_{i = 0}^{n-1} \apdiff{e[e_0/x_0,\ldots,e_{i-1}/x_{i-1}, e_{i +1}/x_{i+1},\ldots,e_{n-1}/x_{n-1}]}{x_i}{e_i} \pdiff{e_i}{x}  \\
\hspace{100pt}  \;\; = \;\; \sum_{i = 0}^{n-1} \pdiff{e}{x_i}[e_0/x_0,\ldots,e_{n-1}/x_{n-1}] 
\pdiff{e_i}{x}
\end{array}
\]
and the conclusion follows. 

In the other direction, assuming the substitution principle, take $e$ to be $\funalpha(x_0,\ldots, x_{n-1})$ (when $\FV(e) \backslash \{x_0,\ldots, x_{n-1}\} = \emptyset$) and $e_i$ to be $\funbeta_i(x)$ for $i = 0, n-1$, to obtain:
\[\vdash \pdiff{\funalpha(\funbeta_0(x), \ldots, \funbeta_{n-1}(x))}{x}
\;\; = \;\; 
\sum_{i = 0}^{n-1} \pdiff{\funalpha(x_0,\ldots, x_{n-1})}{x_i}[\funbeta_0(x)/x_0,\ldots,\funbeta_{n-1}(x)/x_{n-1}] \pdiff{\funbeta_i(x)}{x}\]
As 
\[\pdiff{\funalpha(x_0,\ldots, x_{n-1})}{x_i}[\funbeta_0(x)/x_0,\ldots,\funbeta_{n-1}(x)/x_{n-1}]\]
and
\[\apdiff{\funalpha(\funbeta_0(x),\ldots, \funbeta_{i-1}(x),x_i, \funbeta_{i +1}(x),\ldots, \funbeta_{n-1}(x))}{x_i}{\funbeta_i(x)}\]
are $\realalpha$-equivalent, we are done.
\end{proof}
A \emph{(real) polynomial} is (as usual) an  expression containing no function variables or partial differentiations; it is a \emph{natural number} polynomial if all its constants  are natural numbers. We write 
$P(x_0, \ldots, x_{n-1})$ to indicate that $P$ is  a real polynomial whose (necessarily free) variables are included in $\{x_0, \ldots, x_{n-1}\}$; for expressions  $e_0,\ldots, e_{n-1}$ we may then write  $P(e_0, \ldots, e_{n-1})$ for $P[e_0/x_0, \ldots, e_{n-1}/x_{n-1}]$.
We write $P \sim P'$  to mean that the polynomials $P$ and $P'$ are equal modulo the axioms for commutative rings and the addition and multiplication tables for the real constants.

We remark that our equational system is an instance of a second-order equational logic for binding. This logic has \emph{operator constants} $\op$ of arities of the form $(b_0,\ldots,b_{n-1}; m)$ ($b_i,m \geq 0$). These are used to form a compound expression from $n$ abstracts, of respective arities $b_0,\ldots,b_{n-1}$, and $m$ expressions. With this notation,  $+$ has arity $(;2)$ and $\mathrm{PDiff}$ has arity $(1;1)$.
This system is single-sorted; there is a natural generalisation to a multisorted version. 

In~\cite{FH10} Fiore and Hur gave a multisorted second-order equational logic, and in~\cite{FM10} Fiore and Mahmoud considered the single-sorted case. The differences between their systems and ours are inessential: what we term `function variables' they term `metavariables'; whereas our function variables come with preassigned arities, theirs do not, and they rather utilise suitable environments; their arities are of the simpler form $(b_0,\ldots,b_{n-1})$ but are no less general, as the extra arguments become 0-ary abstracts; and whereas we have both a substitution operation for function variables and a context rule, they employ  a single equivalent rule.  A single substitution version of their equivalent rule is:
\[\frac{e_0 = e_1}{e_0[(x_0,\ldots,x_{n-1}).e'/\funalpha] = e_1[(x_0,\ldots,x_{n-1}).e'/\funalpha]}\]
They employ a simultaneous substitution version of this rule, which is, in any case, equivalent to the single substitution rule.

Returning to our single-sorted version of second-order equational logic, we say that an equational theory is \emph{equationally inconsistent} if $x = y$ is a theorem (with $x$, $y$ different), and we say that an an equational theory is \emph{equationally complete} if any extension of it by a single non-theorem is equationally inconsistent. Note that this is a syntactic criterion, independent of any particular model. 

Polynomials provide an example of equational completeness, though one without a binding operator. The theory is that of commutative rings with constants for all reals and addition and multiplication tables (i.e., our theory less partial differentiation). The idea of the proof is that distinct polynomials have distinct values  for suitable choices of values for their variables. This can be formalised within the logic, and so, given an equation between distinct polynomials, two distinct reals can be proved equal, and that, in turn, enables one to prove $x = y$.

The $\lambda\realbeta\eta$-calculus provides a partial example with a binding operator. It can be formalised as a second-order theory with a binary application operator $\mathrm{ap}\type (;2)$ and a unary lambda abstraction operator $\lambda\type (1;0)$ and the two equations
\[\mathrm{ap}(\lambda ((x).\,\funalpha(x)), y) = \funalpha(y) 
\qquad 
\lambda((y).\, \mathrm{ap}(x,y)) = x\] 
While not equationally complete, it is partially so in the sense that, by B\"{o}hm's theorem~\cite{Bar84}, no two distinct $\realbeta\eta$-normal forms can be consistently equated. 
The theory of the Beta-Bernoulli process~\cite{SSY} also has binding operators; it is formulated using  an equational logic that is a bit different than ours in order to accommodate algebraic effects. 
The theory is almost equationally complete, in the sense that it has just one consistent extension. In this extension all quantitative information about probabilities is lost, and so can be ruled out for the purposes at hand.

\acomment{I have mentally checked assertions on equivalence}

\section{Semantics}

We give a semantics of our partial differentiation expressions using smooth functions $h\type \R^n \rightarrow \R$. 
An \emph{environment}  $\rho$ is a function from variables to reals; 
a \emph{function environment} $\varphi$ is a function from function variables to smooth functions that sends function variables of arity $n$ to smooth functions on $\R^n$. We write $\rho[r_0/x_0,  \ldots, r_{n-1}/x_{n-1}]$ for the environment with value $r_i$ on $x_i$ (for $i = 0,n-1$)  and 
value $\rho(x)$ on any other variable $x$;  function environments 
$\varphi[h_0/\funalpha_0,\ldots, h_{n-1}/\funalpha_{n-1}]$, where $h_i$ is $m$-ary if $\funalpha_i\type m$, are defined similarly.  
We write $0$ for both the constantly $0$ environment and the function environment yielding constantly $0$ functions.

We define the denotation 
$\den{e}\varphi\rho$ of an expression relative to a function environment and an environment by the following clauses:

\[\begin{array}{lcl}
\den{r}\varphi\rho & \;\;  = \;\; & r\\
\den{x}\varphi\rho & \;\;  = \;\; & \rho(x)\\
\den{e + e'}\varphi\rho & \;\;  = \;\; & \den{e}\varphi\rho + \den{e'}\varphi\rho\\
\den{ee'}\varphi\rho & \;\;  = \;\; & \den{e}\varphi\rho \times \den{e'}\varphi\rho\\
\den{\funalpha(e_0,\ldots,e_{n-1})}\varphi\rho & \;\;  = \;\; & \varphi(\funalpha)(\den{e_0}\varphi\rho , \ldots, \den{e_{n-1}}\varphi\rho )\\
\den{\mathrm{PDiff}(x.\, e, e')}\varphi\rho & \;\;  = \;\; & D(r\in \R \mapsto \den{e}\varphi\rho[r/x] )(\den{e'}\varphi\rho )\\
\end{array}\]

Prima facie this definition may not be proper as the last clause makes sense only if the function 
$r\in \R \mapsto \den{e}\varphi\rho[r/x]$  is differentiable. However one proves by structural induction that the definition is proper and, in addition, that for any variables $x_0,\ldots,x_{n-1}$ the function sending
$r_0,\ldots,r_{n-1} \in \R$ to $\den{e}\varphi\rho[r_0/x_0,\ldots,r_{n-1}/x_{n-1}] $ is smooth.

The denotation $\den{e}\varphi\rho$ of an expression depends only on the values ascribed to its free function variables by $\varphi$, and  its free variables by $\rho$. When $e$ has no function variables, we just write $\den{e}\rho$ for its denotation, omitting 
$\varphi$, and when it is closed (has no free variables of either kind) we just write $\den{e}$. 
As usual, we write
$\models e_0 = e_1$
to mean that $e_0$ and $e_1$ have the same denotation, given any function environment and any environment.



We omit the proof of the following standard substitution lemma:

\begin{lemma} \label{semsub}
\lskip
\begin{enumerate}
 \item Variable substitution and denotation commute, that is, we have:
 \[\den{e[e_0/x_0,\ldots, e_{n-1}/x_{n-1}]}\varphi\rho \;\;  = \;\; 
            \den{e}\varphi\rho[\den{e_0}\varphi\rho/x_0,\ldots, \den{e_{n-1}}\varphi\rho/x_{n-1}]\]
 \item  Function variable substitution and denotation commute, that is, we have:
 \[\begin{array}{c}\den{e[(x_0,\ldots,x_{n-1}).e'/\funalpha]}\varphi\rho \;\;  = \;\; 
       \den{e}\varphi[
       u_0,\ldots,u_{n-1}\! \mapsto \den{e'}\varphi\rho[u_0/x_0,\ldots,u_{n-1}/x_{n-1}]
       /\funalpha]\rho
       \end{array}\]
\end{enumerate}
\end{lemma}

Using this lemma, it is then straightforward to prove consistency: \vspace{5pt}
\begin{theorem}[Consistency]\label{consistency} For any expressions $e_0$ and $e_1$ we have:
\[\vdash e_0 = e_1\quad  \implies \quad \models e_0 = e_1 \]
\end{theorem}

Any  polynomial $P(x_0,\ldots,x_{n-1})$ 
defines an $n$-ary  polynomial function on the reals, also written $P(x_0,\ldots,x_{n-1})$. Note that 
\[P(r_0,\ldots,r_{n-1})  \;\;  = \;\; \den{P}0[r_0/x_0,\ldots,r_{n-1}/x_{n-1}]\]
%
  and we recall that $P(x_0,\ldots,x_{n-1}) \sim P'(x_0,\ldots,x_{n-1})$ if, and only if, the  two functions $P(x_0,\ldots,x_{n-1})$ and $P'(x_0,\ldots,x_{n-1})$ are  equal.
 
We say that a function environment is a \emph{(natural number) polynomial} environment if all its values are (natural number) polynomial functions and that an environment is a natural number environment if all its values are natural numbers. We write 
\vspace{0pt}
  \[\models_{\mathbb{N}} e_0 = e_1\]
%
$\mbox{  }$\vspace{0pt}\\
to mean that $e_0$ and $e_1$ have the same denotation, given any natural number polynomial function environment,  and any natural number  environment.

Our semantics of the  theory of partial differentiation employs a concrete notion of function, viz.\  smooth functions over the reals. For the general theory of second-order logic, a more abstract, and thereby more general, notion of function is needed. This can be formulated variously in terms of 
abstract clones~\cite{C81},
Lawvere theories,  or, more conceptually, in terms of monoids in a presheaf category of contexts, see~\cite{FPT99,FH10,FM10}. 

\vspace{25pt}

\section{Canonical forms}

To show completeness we need suitable canonical forms $c$;
up  to a suitable equivalence relation 
$c \approx c'$,
they provide normal forms for our theory.

Let $x_0,x_1,\ldots$ be a countably infinite sequence of distinct variables. For any function variable $\funalpha\type n$, argument sequence $m = i_1,\ldots, i_k \in [n]^*$ (where $[n] = \{0,\ldots,n-1\}$), and expressions $e_0,\ldots,e_{n-1}$ we set:
\[ \funalpha_m(e_0,\ldots,e_{n-1}) \;\;  = \;\;  
\pdiff{^k\!\funalpha(x_0, \ldots, x_{n-1})}{x_{i_1}\ldots \partial x_{i_k}}[e_0/x_0,\ldots, e_{n-1}/x_{n-1}]\]
We note that

\[\FV(\funalpha_m(e_0,\ldots,e_{n-1})) = \bigcup_{i=0}^{n-1} \FV(e_i) \qquad\mbox{     and     }\qquad  \FnV(\funalpha_m(e_0,\ldots,e_{n-1})) = \{f\} \cup \bigcup_{i=0}^{n-1} \FnV(e_i) \] 
and that
\[\funalpha_m(e_0,\ldots,e_{n-1})[e/x] \;\;  = \;\;  \funalpha_m(e_0[e/x],\ldots,e_{n-1}[e/x]) \]
We remark that if two such expressions 
$\funalpha_m(e_0,\ldots,e_{n-1})$ 
and 
$\funalpha'_{m'}(e'_0,\ldots,e'_{n'-1})$ 
are equal then so are: $\funalpha$ and $\funalpha'$;  $m$ and $m'$; $n$ and $n'$; and $e_i$ and $e'_i$, for $i =1,n$; this is used implicitly below to ensure the uniqueness of expression case analyses. We also note that, as may be expected, for any $\varphi$ and $\rho$ we have:
\[\den{\funalpha_m(e_0,\ldots,e_{n-1})}\varphi\rho  =  D_m(\varphi(\funalpha))(\den{e_0}\varphi\rho, \ldots, \den{e_{n-1}}\varphi\rho)\]

The chain rule takes the following form for such applications of partial derivatives of functions:
\begin{lemma} \label{diffap} We have:
\[\vdash \pdiff{\funalpha_m(e_0,\ldots,e_{n-1})}{x} \;\; = \;\; \sum_{i = 0}^{n-1} \funalpha_{im}(e_0,\ldots,e_{n-1})\pdiff{e_i}{x}
\]
\end{lemma}
\begin{proof} Supposing $m = i_1,\ldots, i_k \in [n]^*$, then we have
\[\begin{array}{lcll}\pdiff{\funalpha_m(e_0,\ldots,e_{n-1})}{x} & = & 
 \pdiffop{x}{\left ( \pdiff{^k\!\funalpha(x_0, \ldots, x_{n-1})}{x_{i_1}\ldots \partial x_{i_k}}[e_0/x_0,\ldots, e_{n-1}/x_{n-1}]\right )}\\
         & = & \sum_{i = 0}^{n-1} \left ({  \pdiffop{x}{ \pdiff{^k\!\funalpha(x_0, \ldots, x_{n-1})}{x_{i_1}\ldots \partial x_{i_k}}}}\right )\![e_0/x_0,\ldots, e_{n-1}/x_{n-1}] \, \pdiff{e_i}{x} &\;\; (\mbox{by Lemma~\ref{subdiff}})\\
         & = & \sum_{i = 0}^{n-1} \funalpha_{im}(e_0,\ldots,e_{n-1})\pdiff{e_i}{x}
\end{array}\]
where the second identity is a provable equality, and the other two are syntactic.
\end{proof}

We define a set of \emph{canonical forms} $c$ and a set of \emph{atomic expressions} $a$ (both sets of expressions) by  simultaneous induction: 


\begin{enumerate}
\item 
 \begin{enumerate}
         \item Any  atomic expression is a canonical expression.
         \item Any $r \in \R$ is a canonical expression, as are $c + c'$ and $ cc'$  if $c$ and $c'$ are.
         \end{enumerate}
\item \begin{enumerate}
         \item Any variable $x$ is an atomic expression.
         \item For any $\funalpha\type n$, $m \in [n]^*$,  and canonical forms $c_0,\ldots,c_{n-1}$,  $\funalpha_m(c_0,\ldots,c_{n-1})$ is an atomic expression.
         
         \end{enumerate}
\end{enumerate}
 The set of  \emph{immediate} atomic subexpressions of a canonical form $c$ is defined by structural recursion:
 \[\Imm(a) = \{a\} \qquad \Imm(r) = \emptyset \qquad \Imm(c + c') = \Imm(cc') = \Imm(c) \cup \Imm(c')\]
We remark that the canonical forms are closed under substitution: $c[c'/x]$ is a  canonical form if $c$ and $c'$ are, as is straightforwardly proved by structural induction on $c$. 


\begin{lemma}  \label{onediff}
For any canonical forms $c$ and $\ov{c}$ there is a canonical form $c'$  
such that:
\[\vdash  \apdiff{c}{x}{\ov{c}} = c'\]
and with
$\FnV(c') \subseteq \FnV(c) \cup \FnV(\ov{c})$
and
$\FV(c') \subseteq \FV(c) \cup \FV(\ov{c})$.

%
%
\end{lemma} 
\begin{proof} The proof is by induction on the size of $c$. The cases where $c$ is a constant, variable, sum, or product make use of Lemma~\ref{facts}. 
This leaves the case of an atomic expression
$\funalpha_m(c_0,\ldots,c_{n-1})$, where we can prove:
\[\begin{array}{lcll} \apdiff{\funalpha_m(c_0,\ldots,c_{n-1})}{x}{\ov{c}} & = & \pdiff{\funalpha_m(c_0,\ldots,c_{n-1})}{x} [\ov{c}/x]\\
&= & \left ( \sum_{i = 0}^{n-1} \funalpha_{im}(c_0,\ldots,c_{n-1})\pdiff{c_i}{x} \right ) [\ov{c}/x]&(\mbox{by Lemma~\ref{diffap}}) \\
&= &  \sum_{i = 0}^{n-1} \funalpha_{im}(c_0 [\ov{c}/x],\ldots,c_{n-1} [\ov{c}/x])\apdiff{c_i}{x}{\ov{c}}\\
\end{array}\]
and recall  the canonical forms are closed under substitution and apply the induction hypothesis to the $c_i$.
\end{proof}

\begin{lemma}[Canonicalisation] \label{canonicalisation} For any expression    $e$ there is a canonical form $\CF(e)$ 
 such that $\vdash e = \CF(e)$ and with $\FnV(\CF(e)) \subseteq \FnV(e)$ and $\FV(\CF(e)) \subseteq \FV(e)$.
\end{lemma}
\begin{proof} We prove this by structural induction on $e$. 
If $e$ is a variable or a constant it is already a canonical form.
The cases where $e$ is a sum or a product follow from the fact that the canonical forms are closed under sums and products. 
The case where $e$ has the form $\funalpha(e_0,\ldots, e_{n-1})$ is immediate from the induction hypothesis.
Finally, the case where $e$ has the form $ \apdiff{e_0}{x}{e_1}$ is handled using Lemma~\ref{onediff}, and the induction hypothesis.
\end{proof}
We have the following corollary of the canonicalisation lemma:
\begin{corollary} For any closed expression $e$, we have:
\[\vdash e = \den{e}\]
\end{corollary} \label{closed}
\begin{proof} As $e$ is closed, the same, by Lemma~\ref{canonicalisation}, is true of $\CF(e)$, which latter must therefore have no immediate atomic subexpressions or free variables, and so is a closed polynomial. There is therefore a real $r$ such that $\vdash \CF(e) = r$ and so, by the lemma, $\vdash e = r$. By consistency we have $r = \den{e}$.
\end{proof}

We next define the equivalence relations that will allow us to
 characterise provable equality between canonical forms. We also define 
  polynomials that will  be used to formulate our Hermite interpolation problems.

For any $m,m' \in [n]^*$ we write $m \sim m'$ to mean that $m$ is a permutation of $m'$.  
For any atomic expressions $a, a'$ we write $a \approx a'$ to mean that 
\begin{itemize}
\item $a$ and $a'$ are identical variables, or else 
\item
for some function variable $\funalpha\type n$, and for some  $m \sim m'$,  they have the forms $\funalpha_m(c_0,\ldots,c_{n-1})$ and $\funalpha_{m'}(c'_0,\ldots,c'_{n-1})$ with $\vdash c_i = c'_i$ for $i = 0,n-1$.
\end{itemize}
We evidently have:
\begin{equation} \label{a=} a \approx a' \implies \vdash a = a' \hfill \end{equation}

We now fix  an assignment $v_a$ of variables (termed \emph{separation} variables) to atomic expressions such that:
\[v_a = v_{a'} \iff a \approx a'\]
%
For any canonical form $c$, we define its \emph{node}  polynomial $P_c$ by structural induction:
\[P_r = r\; (r \in \R) \qquad  P_a = v_a \; 
\qquad P_{c + c'} = P_c + P_{c'} \qquad P_{c c'} = P_c P_{c'}\]
The separation variables occurring in $P_c$ are the $v_a$ 
with $a$ an immediate atomic subexpression of $c$. 

%
%

We then define our equivalence relation between canonical forms by:
\[c \approx c' \quad  \iff \quad P_c \sim P_{c'}\]
Note that this relation only depends on the choice of the $v_a$ for the immediate atomic expressions of $c$ and $c'$.
Note too that the relation extends that between atomic expressions, i.e.,  $a \approx a'$ holds with $a, a'$ taken as atomic expressions iff it does when they are taken as canonical forms. 
%



\begin{lemma} \label{ae}
\lskip
 \begin{enumerate} \item For any canonical expression $c$ with node polynomial $P_{c}(v_{a_0},\ldots,v_{a_{m-1}})$, for atomic expressions $a_0,\ldots, a_{m-1}$, we have:
\[\vdash c = P_{c}(a_0,\ldots,a_{m-1})\]
%
\item For any canonical expressions $c$ and $c'$ we have:
\[c \approx c' \implies \vdash c = c'\]
\end{enumerate}
\end{lemma}
\begin{proof} 
The first part  is established by a straightforward structural induction. In the case where $c$ is an atomic expression $a$, we have $v_a = P_c = v_{a_0}$ and so $a \approx a_0$ and so $\vdash a = a_0$, i.e., $\vdash c = P_c(a_0)$.

 For the second part, suppose that we have canonical expressions $c$ and $c'$ such that $c \approx c'$ (and so $P_c \sim P_{c'}$).
Let $v_{a_0},\ldots,v_{a_{m-1}} \; (j = 0,m)$ include all the variables of $P_c$ and $P_{c'}$.  As $P_c \sim P_{c'}$, we have:
\[\vdash P_c[a_0/v_{a_0},\ldots,a_{m-1}/v_{a_{m-1}}] = P_{c'}[a_0/v_{a_0},\ldots,a_{m-1}/v_{a_{m-1}}]\]
Applying  the first part to $c$ and $c'$, it  follows that $\vdash c = c'$, as required.
\end{proof}

\section{Completeness}

To establish completeness we use environments distinguishing non-equivalent canonical forms, i.e., canonical forms $c$ and $c'$ such that  $c \not\approx c'$. We first need  to be able to distinguish finite sets of non-equivalent polynomials  (these will be the node polynomials of sub-canonical forms of $c$ or $c'$):
%
%
%
\begin{lemma} \label{poly-dif}  Let $P_i(x_0,\ldots, x_{m-1})$ be $n$  mutually inequivalent  polynomials. 
Then, for some natural number choices of   $x_0,\ldots, x_{m-1}$, they take on different values.
\end{lemma}
\begin{proof} 
None of the polynomials $Q_{ij} = P_i - P_j$ are identically $0$ ($i \neq j$), and so neither is $Q = \Pi_{ij}Q_{ij}$. So $Q$ is non-zero for some natural number choices of   $x_0,\ldots, x_{m-1}$. That choice differentiates  distinct $P_i$. 
%
\end{proof}

We next need to be able to solve multivariate Hermite interpolation problems
in order to give prescribed values to atomic expressions.
A \emph{multivariate Hermite interpolation problem of dimension $d \geq 0$} is given by: 
\begin{itemize}
\item A finite set of  \emph{nodes} $x_i \in \R^d$ ($i=1.k$), and
\item  For each node $x_i$, finitely many \emph{conditions}  $D_{m_{ij}}(h)(x_i) = r_{ij}$ 
($m_{ij} \in [d]^*, r_{ij} \in \R$) on the partial derivatives of a function $h\type \R^d \rightarrow \R$. The set of conditions must be \emph{consistent}, in the sense that each $r_{ij}$ is determined by the node $x_i$ and the  permutation equivalence class of $m_{ij}$.
\end{itemize} 
A \emph{solution} to this problem is a $d$-ary polynomial function satisfying all the conditions. 
According to a theorem of Severi~\cite{Sev21} (and see~\cite{Lor92,Lor00}) every such problem has a solution by a polynomial of degree $k(\max(|m_{ij}|) + 1) -1$.

\acomment{In above it is assumed that the $D_{m_{ij}}(f)$ notation is standard}

A finite set $C$ of canonical forms is \emph{saturated} if the following two conditions hold:
\begin{enumerate}
\item For any canonical form $c$ and immediate atomic subexpression $a$ of $c$
\[c \in C \; \implies \; a \in C \]
\item \[\funalpha_m(c_0,\ldots,c_{n-1}) \in C \; \implies \; c_0,\ldots,c_{n-1} \in C\]
\end{enumerate}


Every finite set of canonical forms can evidently be extended to a finite saturated set of canonical forms.

\begin{theorem}[Polynomial Separation] \label{sep} 
Let $C$ be a finite set of canonical forms. Then there is a polynomial function environment  $\varphi$ and an environment $\rho$ that distinguish any two inequivalent elements of  $C$.
\end{theorem} 
\begin{proof}
We can assume w.l.o.g.\ that $C$ is saturated. 
By Lemma~\ref{poly-dif} we can find an assignment $\ov{\rho}$ of reals to the separation variables of the node  polynomials 
$P_c$,  ($c \in C$) 
which separates them, i.e., makes inequivalent ones take on different values, $r_c$.
As 
$P_c \sim P_{c'}$ iff $c \approx c'$, we have $c \approx c'$ iff $r_c = r_{c'}$. 
In the case that $c$ is an atomic expression $a$ we have 
$P_a = v_a$ and so $r_a = \ov{\rho}(v_a)$.
%
We  define an environment $\rho$ by:
\[\rho(x) \;\;  = \;\; \left \{ \begin{array}{ll}
                                  r_x & (x \in C)\\
                                  0 & (\mbox{otherwise})
                           \end{array}
                  \right .\]



We next set up an $n$-dimensional Hermite interpolation problem for each function variable $\funalpha\type n$ appearing
 in  the atomic expressions in $ C$.
For each atomic expression $\funalpha_m(c_0,\ldots,c_{n-1}) \in C$ we add a  node
 \[(r_{c_0},\ldots,r_{c_{n-1}})\] 
 and a condition:
  \[D_m(h)(r_{c_0},\ldots,r_{c_{n-1}}) = r_{\funalpha_m({c_0},\ldots, c_{n-1})}\] 
Note that the nodes are determined by the values of the node polynomials.

To show these conditions consistent, let $a = \funalpha_m(c_0,\ldots,c_{n-1})$ and $a' = \funalpha_{m'}(c'_0,\ldots,c'_{n-1})$ be two atomic expressions in $C$ such that $m \sim m'$ and $r_{c_i} = r_{c'_i}$ for $i = 0,n-1$.  Then, by the above, we have 
$c_i \approx c'_i$, and so,  by  part (2) of Lemma~\ref{ae}, $\vdash c_i = c'_i$, for $i = 0,n-1$. So $a \approx a'$ as atomic expressions and so too, therefore, as canonical forms. It follows that $r_a = r_{a'}$, as required for consistency.

We can therefore obtain a polynomial function environment by taking $\varphi(\funalpha)$ to be an $m$-ary interpolating polynomial for the corresponding Hermite interpolation problem, if $\funalpha$ appears in some $c \in C$, and constantly $0$, otherwise.



We claim that:
 \[\den{c}\varphi\rho \;\;  = \;\; r_c\]
  for every $c \in C$.
We establish the claim by structural induction on the expressions in $C$.

First, consider a variable $x \in C$. 
We have $\den{x}\varphi\rho   = \rho(x) = r_x$, as required.

Next, consider an atomic expression  $\funalpha_m(c_0,\ldots,c_{n-1}) \in C$. By saturation we have $c_0,\ldots,c_{n-1} \in C$.  So, by the induction hypothesis, we have  $\den{c_i}\varphi\rho = r_{c_i}$ for $i = 0, n-1$. We then have:
\[\begin{array}{lcl}\den{\funalpha_m(c_0,\ldots,c_{n-1})}\varphi\rho & = & D_m(\varphi(\funalpha))(\den{c_0}\varphi\rho, \ldots, \den{c_{n-1}}\varphi\rho)\\
 & \;\;=\;\; & D_m(\varphi(\funalpha))(r_{c_0}, \ldots, r_{c_{n-1}})\\
 & \;\;=\;\; &  r_{\funalpha_m({c_0},\ldots, c_{n-1})}
\end{array}\]
with the last line holding as $\varphi(\funalpha)$ solves the Hermite interpolation problem  for $\funalpha$ set up above.

Lastly, consider a  canonical form $c \in C$ which is not an atomic expression. Let $v_{a'_0}, \ldots, v_{a'_{n-1}} $ be the variables of $P_c$. For every $a'_i$ there is an immediate atomic subexpression $a_i$ of $c$ such that $a_i \approx a'_i$ (and so $v_{a'_i} = v_{a_i}$).
By  part (1) of Lemma~\ref{ae} we then have
\[\vdash c = P_c[a_{0}/v_{a_0}, \ldots, a_{n-1}/v_{a_{n-1}}]\]
%
 %
 As $C$ is saturated, it contains the $a_i$. Since $c$ is not itself an atomic expression, we can apply the induction hypothesis to the $a_i$ and so we have $\den{a_i}\varphi\rho = r_{a_i}$.
 
 We may then calculate that:
\[\begin{array}{lcl}
\den{c}\varphi\rho & \;\;=\;\;  &\den{P_c[a_0/v_{a_0}, \ldots, a_{n-1}/v_{a_{n-1}}]}\varphi\rho\\
                             & \;\;=\;\;  &\den{P_c(v_{a_0}, \ldots, v_{a_{n-1}})}\varphi
                                     \rho[\den{a_0}\varphi\rho/v_{a_0}, \ldots, 
                                            \den{a_{n-1}}\varphi\rho/v_{a_{n-1}}] \\
                             & \;\;=\;\;  &\den{P_c(v_{a_0}, \ldots, v_{a_{n-1}})}\varphi
                                          \rho[r_{a_0}/v_{a_0}, \ldots, r_{a_{n-1}}/v_{a_{n-1}}] \\
                              & \;\;=\;\;  &\den{P_c(v_{a_0}, \ldots, v_{a_{n-1}})}\varphi
                                          \rho[\ov{\rho}(v_{a_0})/v_{a_0}, \ldots, \ov{\rho}(v_{a_{n-1}})/v_{a_{n-1}}] \\
                              & \;\;=\;\; & r_c
\end{array}\]
with the last line holding by the definition of $r_c$. This concludes the inductive  proof.


Finally, taking  $c \not\approx c' \in C$, we have $r_c \neq r_{c'}$, and so, by the above, $\den{c}\varphi\rho \neq \den{c'}\varphi\rho$, and so, as required, $\varphi$ and $\rho$ distinguish any two inequivalent elements of  $C$, concluding the proof.
\end{proof}

%
%

We next strengthen the polynomial separation theorem to natural number separation:


\begin{theorem}[Natural Number Separation] \label{natsep} 
Let $C$ be a finite set of canonical forms. Then there is a natural number polynomial function environment  $\varphiN$ and a natural number environment $\rhoN$ that distinguish any two inequivalent elements of  $C$.
\end{theorem} 
\begin{proof}
 By  polynomial separation (Theorem~\ref{sep}), there is a polynomial function environment $\varphi$ and an environment $\rho$ such that distinguish any two inequivalent elements of $C$.

Let $\funalpha_i \type m_i$ (for $i = 0,  n-1$) be the function variables occurring in the $c \in C$, and, for $i = 0,  n-1$,  let 
$P_i(x_0, \ldots, x_{m_i -1})$ be polynomials defining $\varphi(\funalpha_i)$.
Let $r_0, \ldots , r_{q-1}$ be  the constants occurring in these $P_i$ and, for $i = 0,  n-1$, let $Q_i(x_0, \ldots, x_{m_i -1},  y_0,  \ldots, y_{q-1})$ be natural number polynomials such that
\[Q_i[r_0/y_0,  \ldots, r_{q-1}/y_{q-1}] = P_i\]
and where no $y_k$ occurs freely in any $c \in C$.  
(The $Q_i$ can be obtained   from the $P_i$ by replacing $r_k$ by $y_k$, for $k = 0,q-1$.)

For $c \in C$, set $d_c = \CF(c[\ldots, (x_0,  \ldots, x_{m_i}).\, Q_i/\funalpha_i ,\ldots])$.  As $\funalpha_i$ is a list of function variables including all those of $c$, and as $\FnV(c) \subseteq \FnV(e)$ by Lemma~\ref{canonicalisation}, we see that $c[\ldots, (x_0,  \ldots, x_{m_i}).\, Q_i/\funalpha_i ,\ldots]$ has no function variables and so, again by Lemma~\ref{canonicalisation},  neither does its canonical form $d_c$. As any canonical form with no function variables is a polynomial, it follows that $d_c$ is a polynomial. 
 
 Setting $\ov{\rho} = \rho[\ldots,r_k/y_k,\ldots]$, we calculate:
\[\hspace{-8pt}\begin{array}{lcl}
\den{c}{\varphi}{\rho} & \;\;=\;\; &   
  \den{c}{\varphi[\ldots, (u_0,\ldots,u_{m_i}\! \in\! \mathbb{R}\mapsto \den{P_i}\rho[\ldots, u_j/x_j,\ldots])/\funalpha_i,\ldots]}{\rho}\\[0.3em]

& \;\;=\;\; & \den{c}
          {\varphi[\ldots, (u_0,\ldots,u_{m_i}\! \in\! \mathbb{R}\mapsto \den{Q_i[ \ldots, r_k/y_k, \ldots]}\rho[\ldots, u_j/x_j,\ldots])/\funalpha_i,\ldots]}{\rho}\\[0.3em]

& \;\;=\;\; & \den{c}
          {\varphi[\ldots, (u_0,\ldots,u_{m_i}\! \in\! \mathbb{R}\mapsto \den{Q_i}\rho[\ldots, u_j/x_j,\ldots,r_k/y_k,\ldots])/\funalpha_i,\ldots]}{\rho}\\[0.3em]

& \;\;=\;\; & \den{c}
          {\varphi[\ldots, (u_0,\ldots,u_{m_i}\! \in\! \mathbb{R}\mapsto \den{Q_i}\ov{\rho}[\ldots, u_j/x_j,\ldots])/\funalpha_i,\ldots]}{\ov{\rho}}\\[0.3em]

& \;\;=\;\; & \den{c[\ldots, (x_0,  \ldots, x_{m_i}).\, Q_i/\funalpha_i ,\ldots]}
          {\varphi}{\ov{\rho}}\\[0.3em]

& \;\;=\;\; & \den{d_c}{\ov{\rho}}\\[0.3em]
 \end{array}\]
 where we have used Lemma~\ref{semsub}, canonicalisation, and consistency (Theorem~\ref{consistency}). 
As $\varphi$, $\rho$ separate inequivalent $c \in C$ we therefore see that the corresponding $d_c$ are inequivalent as functions and so as polynomials. By Lemma~\ref{poly-dif} there is  a natural number environment $\rhoN$ separating any two inequivalent $d_c$, and so any $d_c$ and $d_c'$ with $c$ and $c'$ inequivalent.

 
 Next, define a natural number polynomial function environment $\varphiN$ by setting:
 \[\varphiN(\funalpha) \;\; = \;\; 
     \left \{ \begin{array}{cl} u_0, \ldots, u_{n_i} \in \R \mapsto  \den{Q_i}\rhoN[u_0/x_0,  \ldots, u_{n_i}/x_{n_i}] & (\funalpha = \funalpha_i)\\
                                              0 & (\mbox{otherwise})
                                                                                                                                                                \end{array}\right .\]
For any $c \in C$ we have:
 \[\begin{array}{lcl}
 \den{c}\varphiN\rhoN 
                                   & \;\;=\;\; &  \den{c}0[\ldots, (u_0, \ldots, u_{n_i} \in \R \mapsto  \den{Q_i}\rhoN[\ldots, u_j/x_j,  \ldots])/\funalpha_i, \ldots]\rhoN  \\[0.3em]
                                   & \;\;=\;\; & \den{c[\ldots, (x_0,  \ldots, x_{m_i}).\, Q_i/\funalpha_i ,\ldots]}0\rho_N     \\[0.3em]
                                   & \;\;=\;\; & \den{d_c}\rho_N
 \end{array}\]
again making use of  Lemma~\ref{semsub}, canonicalisation, and consistency. 

So as $\rho_N$ separates $d_c$ and $d_{c'}$  whenever $c$ and $c'$ are inequivalent, we see that the natural number polynomial function environment $\varphiN$ and the natural number environment $\rhoN$ separate any two inequivalent elements of $C$, concluding the proof.
\end{proof}

We can now prove completeness relative to natural number polynomial function environments and natural number environments:
\begin{theorem}[Natural number completeness] \label{complete} For any expressions $e$ and $e'$ we have:
\[\models_{\mathbb{N}} e = e' \implies \vdash  e = e'\]
\end{theorem}
\begin{proof} Suppose that $\models_{\mathbb{N}} e = e'$ but, for the sake of contradiction, that $\not\vdash e = e'$. 
By canonicalisation (Lemma~\ref{canonicalisation}) we then have 
$\not\vdash c = c'$, where $c = \CF(e)$ and $c' = \CF(e')$; so, by part (2) of Lemma~\ref{ae}, we have $c \not\approx c'$.
Therefore, by  
Theorem~\ref{natsep}, there is a natural number polynomial function environment $\varphiN$ and a natural number environment $\rhoN$ such that
$\den{c}\varphiN\rhoN \neq \den{c'}\varphiN\rhoN$, obtaining the required contradiction.
\end{proof}

We next upgrade part (2) of Lemma~\ref{ae} to an equivalence and use that to analyse the  theorems of our theory  in terms of canonical forms and their $\approx$ relation; we also obtain  a useful characterisation of the equivalence of atomic expressions in terms of the equivalence of their  canonical subexpressions:
\begin{theorem} \label{analysis}
\lskip
\begin{enumerate}
\item For any expressions $e$ and $e'$ we have:
\[\vdash e = e' \iff \CF(e) \approx \CF(e')\]
\item For any canonical forms  $c$ and $c'$ we have:
\[\vdash c = c' \iff c \approx c'\]
\item For any atomic  forms  $a$ and $a'$, $a \approx a'$ holds iff:
\begin{enumerate}
\item $a$ and $a'$ are identical variables, or else 
\item
for some function variable $\funalpha\type n$, and for some  $m \sim m'$,  they have the forms $\funalpha_m(c_0,\ldots,c_{n-1})$ and $\funalpha_{m'}(c'_0,\ldots,c'_{n-1})$ with $ c_i \approx c'_i$ for $i = 0,n-1$.
\end{enumerate}
\end{enumerate}
\end{theorem}
\begin{proof} For the second part, we already have the implication from right to left, by   part (2) of Lemma~\ref{ae}. The other direction follows from consistency and separation. For the first part, by canonicalisation we have that $\vdash e = e'$ holds iff $\vdash \CF(e) = \CF(e')$ does, and
%
%
 conclusion follows from the second part. 
The third part follows immediately from the second part and the definition of the equivalence relation between atomic expressions.
\end{proof}

We next seek a local criterion for canonical expression equivalence. Let $A$ be a set of atomic expressions, and let $w_a \;\; (a\in A)$ be an $A$-indexed set of variables. We define polynomials $P_{w,c}$ for  canonical forms whose immediate atomic subexpressions are included in $A$ by structural induction on $c$\hspace{0.2pt}:
\[P_{w,r} = r\; (r \in \R) \qquad  P_{w,a} = w_a \; 
\qquad P_{w,c + c'} = P_{w,c} + P_{w,c'} \qquad P_{w,c c'} = P_{w,c} P_{w,c'}\]
and we say that $w$ is \emph{equivalence-characterising} if, for all $a,a' \in A$, we have:    
\[w_a = w_{a'} \iff a \approx a'\]
%
\begin{lemma} \label{can-anal} Let $c,c'$ be canonical forms, and let $A$ be a set of atomic expressions including the immediate atomic subexpressions of $c$ and $c'$ and let $w_a \;\; (a\in A)$ be an equivalence-characterising $A$-indexed set of variables. Then:
\[c \approx c' \iff  P_{w,c} \sim P_{w,c'}\]
\end{lemma}
\begin{proof} It suffices to show that $P_c \sim P_{c'}$ holds iff $P_{w,c} \sim P_{w,c'}$ does. Choose $a_1,\ldots,a_k \in A$ such that $w_{a_1},\ldots,w_{a_k}$ enumerates $\{w_a\, |\,  a \in A\}$. Then, as both $w$ and $v$ are equivalence-characterising, $v_{a_1},\ldots,v_{a_k}$ enumerates $\{v_a | a \in A\}$. A straightforward induction shows that $P_c[w_{a_1}/v_{a_1},\ldots,w_{a_k}/v_{a_k}] = P_{w,c}$, and similarly for  $c'$. Assuming 
$ P_{c} \sim P_{c'}$, we then find:
\[P_{w,c} = P_c[w_{a_1}/v_{a_1},\ldots,w_{a_k}/v_{a_k}] \sim P_{c'}[w_{a_1}/v_{a_1},\ldots,w_{a_k}/v_{a_k}] = P_{w,c'}\]
that is, $ P_{c} \sim P_{c'}$ implies $ P_{w,c} \sim P_{w,c'}$. The converse is proved similarly.
\end{proof}

As we now show, the equivalence relation on canonical forms can be viewed as a combination of polynomial equivalence and the commutativity of partial differentiation with respect to different variables. We write
\[ \rdash e = e'\]
to mean that $e $ and $e'$ can be proved equal using only the ring axioms, the addition and multiplication tables, and the commutativity axiom for partial differentiation (i.e., without using any partial differentiation axioms other than commutativity).
\begin{lemma} \label{technical}
For any canonical form $c$, and set of atomic expressions $A$ with $\Imm(c) \subseteq A$ 
and such that for all $a, a' \in A$ if $a \approx a'$ then $\rdash a = a'$, and for 
any $\approx$-characterising $A$-indexed variable assignment $w$ we have 
\[\rdash c =  P_{w,c}[a_1/v_{a_1}, \ldots, a_n/v_{a_n}]\]
where $v_{a_1}, \ldots, v_{a_n}$ are the variables of $P_{w,c}$ for  $a_1, \ldots, a_n \in A$.
\end{lemma}
\begin{proof}
The proof is a simple structural induction on $c$. For
an atomic expression $a$, we have $P_{w,c} = v_a$ and $v_a = v_{a_1}$,  by assumption. So as $w$ is $\approx$-characterising, we have $a \approx a_1$ and so $\rdash a = a_1$, by assumption.
\end{proof}

\begin{theorem} \label{RTC} For any canonical forms $c$ and $c'$ we have:
\[\vdash c = c' \iff c \approx c' \iff \rdash c = c'\]
\end{theorem}
\begin{proof} We prove that
for all $c,c'$
\[c \approx c' \implies \rdash c = c'\]
by induction on $\max(|c|,|c'|)$, when the conclusion follows immediately from Theorem~\ref{analysis}.

In the case where $c$ (say)  is an atomic expression $a$, then $P_{c'} \sim P_c = v_a$, and so $\rdash c' = a'$ for some $a' \in \Imm(c')$ with $v_{a'} = v_a$.
As $v_{a'} = v_a$, we have $a \approx a'$, and we then use the induction hypothesis to prove $\rdash a  = a'$ using  the characterisation of equivalence of atomic expressions given by part (3) of  Theorem~\ref{analysis}.  

Otherwise neither $c$ nor $c'$ is an atomic expression and so, setting $A = \Imm(c) \cup \Imm(c')$, by induction  
 we have $\rdash a = a'$ whenever $a \approx a'$, for $a,a' \in A$. Let $w$ be any $\approx$-characterising $A$-indexed variable assignment. Applying Lemma~\ref{technical} we find 
$\rdash c =  P_{w,c}[a_1/v_{a_1}, \ldots, a_n/v_{a_n}]$ 
and 
$\rdash c' =  P_{w,c'}[a_1/v_{a_1}, \ldots, a_n/v_{a_n}]$.
Further, by Lemma~\ref{can-anal} , we have $P_{w,c} \sim P_{w,c'}$. It follows that
$\rdash P_{w,c}[a_1/v_{a_1}, \ldots, a_n/v_{a_n}] = P_{w,c'}[a_1/v_{a_1}, \ldots, a_n/v_{a_n}]$. Putting these three things together we obtain $\rdash c = c'$, as desired.
\end{proof}

As may be evident, one can further strengthen the definition of canonical forms so that fewer axioms are needed to prove equivalence. For example, the partial differentiations in atomic expressions can be put in a standard order using the commutativity axiom. One can  show that two such canonical forms are provably equal if, and only, if they can be proved so without using any axioms for partial differentiation; the proof of this fact parallels that of Theorem~\ref{RTC}. Should one wish, one can go further and rewrite polynomials as sums of distinct multinomials; two such canonical forms are provably equal if, and only if, they can be proved so using only the associativity and commutativity of $+$.

We turn next to equational completeness.
Let $\funalpha_i\type m_i \; (i = 1,k)$ and   $x_j \; (j = 1,l)$ be the  function variables and the free variables of two expressions $e$ and $e'$. Then a \emph{natural number counterexample} to the equation $e=e'$  consists of
 natural number polynomials $P_i(y_0,  \ldots, y_{m_i})$ and natural numbers $k_1,\ldots, k_l$ such that,
 the closed terms  $\underline{e} \eqdef e[\ldots, (y_0,  \ldots, y_{m_i}).\, P_i/\funalpha_i ,\ldots][\ldots k_j/x_j, \ldots ]$ and $\underline{e}' \eqdef e'[\ldots, (y_0,  \ldots, y_{m_i}).\, P_i/\funalpha_i ,\ldots][\ldots k_j/x_j, \ldots ]$
 are provably equal to different constants.
 %
 As we now see, the natural number separation theorem yields such counterexamples, and thereby enables us to establish equational completeness:
 %

\begin{theorem} \label{HP}The theory of partial differentiation is equationally complete.
\end{theorem}
\begin{proof} Suppose that $\not\vdash e = e'$.  We show that adding the equation $e = e'$ makes the theory equationally inconsistent. As $\not\vdash e_0 = e_1$, by natural number polynomial completeness (Theorem~\ref{complete}), there is a natural number polynomial function environment $\varphi$, and a natural number environment $\rho$ such that $\den{e}\varphi\rho$ and $\den{e'}\varphi\rho$ differ.
We can then define a natural number counterexample to the equation. Let $\funalpha_i\type m_i \; (i = 1,k)$ and   $x_j \; (j = 1,l)$ be the  function variables and the free variables of $e$ and $e'$.
For $i = 1, k$, take 
$P_i(y_0,  \ldots, y_{m_i})$ to be a natural number polynomial defining $\varphi(\funalpha_i)$ and take $k_j$ to be $\rho(x_j)$, and define $\underline{e}$ and $\underline{e'}$ as above. By the substitution rule we have $\vdash \underline{e} = \underline{e'}$. By Lemma~\ref{semsub},  we have $\den{\underline{e}} = \den{e}\varphi \rho$ and $\den{\underline{e}'} = \den{e'}\varphi \rho$. So $\den{\underline{e}} \neq \den{\underline{e}'}$, and therefore by Corollary~\ref{closed} and consistency, $\underline{e}$  and $\underline{e}'$ are provably equal to different constants.


Thus, assuming $\vdash e = e'$, we can prove distinct real constants 
$\den{\underline{e}}$ and $\den{\underline{e}'}$ equal. Using the addition and multiplication tables and the ring axioms, we then find that 
$\vdash 1 = 0$ and so
  $ \vdash  x = 0$ and so 
   $\vdash  x = y$. 
\end{proof}

Instead of allowing all reals as constants, one could restrict them, for example  to $\mathbb{Q}$ or   none, except for the ring constants (equivalently $\mathbb{N}$). 
The above development goes through straightforwardly in either of these  cases, except for equational completeness. The proof of equational completeness for $\mathbb{Q}$ goes through as the required polynomials, being over the natural numbers, and the required reals, being natural numbers, are definable, and the argument from the equality of two distinct rationals to that of two distinct variables also goes through.

However, the last part of that argument does not go through  for $\mathbb{N}$. For example it is then consistent to add the equation $1 + 1 = 0$. As a model, one can employ the \emph{boolean differential calculus}~\cite{PS04}. One works over the boolean ring and defines the partial differentiation of boolean functions $h \type \B^n \rightarrow \B$ by:
\[\pdiff{h(x_0, \ldots, x_{n-1})}{x_i} = h(x_0, \ldots, x_{i-1},0, x_{i+1},\ldots, x_{n-1}) + h(x_0, \ldots, x_{i-1},1, x_{i+1},\ldots, x_{n-1})\]
This can be equivalently written, somewhat more transparently, as:
\[\pdiff{h(x_0, \ldots, x_{n-1})}{x_i} = h(x_0, \ldots, x_{i-1},x_i + \mathrm{d} x_i, x_{i+1},\ldots, x_{n-1}) - h(x_0, \ldots, x_{i-1},x_i, x_{i+1},\ldots, x_{n-1})\]
taking $\mathrm{d} x_i = 1$.
Another model of $1 + 1 = 0$  can be constructed from   the clone of polynomials in several variables over the boolean ring, with partial differentiation defined as usual on polynomials over a ring.


With the constants restricted to the rationals, the question of decidability of the equational theory of partial differentiation makes sense, and we have:
\begin{theorem}\label{decidable} With the constants restricted to the rationals, the equational theory of partial differentiation is decidable, and, further,  natural number counterexamples to unprovable equations can be effectively obtained.
\end{theorem}
\begin{proof}
We interleave two search procedures. One searches for a proof, the other searches through all possible natural number counterexamples.
Since we know from the proof of equational completeness (Theorem~\ref{HP}) that a natural number counterexample exists if an equation is not provable, this procedure will terminate, yielding either a proof or a counterexample.
\end{proof}

There is another, more direct,  way to prove decidability and to find counterexamples. For decidability,  one uses the effectiveness of the proof of canonicalisation (Lemma~\ref{canonicalisation}) to find canonical forms, and  the characterisations of the equivalence of atomic expressions and canonical forms given by Theorem~\ref{analysis} and Lemma~\ref{can-anal}, to obtain  a recursive algorithm to decide the equivalence of canonical forms. 

With this decision procedure in hand, and with the observation that the proof in~\cite{Lor92} of the Severi theorem  is effective, one observes that the proof of polynomial separation (Theorem~\ref{sep}) is effective. So too, therefore is the proof of natural number separation (Theorem~\ref{natsep}), and, following the proof of equational completeness (Theorem~\ref{HP}), we finally see that, if they exist, 
counterexamples can be found effectively.

\section*{Acknowledgements} I thank Mart\'{i}n Abadi, Jonathan Gallagher, and Tarmo Uustalu for useful discussions, and Sam Staton for helpful remarks including indicating how the  axioms for partial differentiation could be made finitary.


\begin{thebibliography}{10}\label{bibliography}
      
 \bibitem{Bar84} Barendregt, Hendrik P., ``The Lambda Calculus: Its Syntax and Semantics", Studies in Logic and the Foundations of Mathematics, \textbf{103}, 2nd edition, North-Holland, 1985.
 
 \bibitem{BCS09} Blute, Richard F., Robin Cockett, and Robert A.G. Seely,  \emph{Cartesian Differential Categories}, Theory and Applications
of Categories, {\bf 22} (2009), 622--672.

 \bibitem{CC14} Cockett, Robin and Geoff Cruttwell, \emph{Differential Structure, Tangent Structure, and
SDG}, Applied Categorical Structures,  {\bf 22}(2) (2014), 331--417.


      \bibitem{CCG20}	Cockett, Robin, Geoff Cruttwell, Jonathan Gallagher, Jean-Simon Pacaud Lemay,  Benjamin MacAdam, Gordon Plotkin,  and Dorette Pronk,
\emph{Reverse Derivative Categories}, Proc.\ 28th.  Annual Conference on Computer Science Logic  
(eds.\ Maribel Fern{\'{a}}ndez and Anca Muscholl), 
LIPIcs, {\bf 152}, 18:1-18:16, Schloss Dagstuhl, 2020.

\bibitem{C81} Cohn, Paul Moritz, ``Universal algebra", Mathematics and Its Applications, {\bf 6}, 2nd edn., Dordrecht: Reidel, 1981.


 \bibitem{FH10} Fiore, Marcelo and Chung-Kil Hur, \emph{Second-order equational logic}, 
 {Proc.\ 24th.\  International Workshop on Computer Science Logic} (eds.\ Anuj Dawar and Helmut Veith),  LNCS {\bf 6247}, 320--335, Springer, 2010.
 
 
 \bibitem{FM10} Fiore, Marcelo and Ola Mahmoud,
 \emph{Second-order algebraic theories}, 
 {Proc.\ 35th.\ International Symposium on Mathematical Foundations of Computer Science} 
 (eds.\  Hlin\v{e}n\'{y} P., Ku\v{c}era A.), LNCS \textbf{6281}, 368--380,  Springer, 2010.
    

 \bibitem{FPT99} Fiore, Marcelo, Gordon Plotkin, and Daniele Turi, 
 \emph{Abstract syntax and variable binding},
 {Proc.\ 14th.\  Annual IEEE Symposium on Logic in Computer Science}, 193--202, 1999.
 

 \bibitem{Lor92} Lorentz, Rudolph~A.,  ``Multivariate Birkhoff Interpolation",  Lecture Notes in Mathematics, \textbf{1516},  Springer, 1992.

\bibitem{Lor00} Lorentz, Rudolph~A., \emph{Multivariate Hermite interpolation by algebraic polynomials: A survey},
 {Journal of Computational and Applied Mathematics}, \textbf{122} (2000), 167--201.

\bibitem{PS04} Posthoff, Christian, and Bernd Steinbach, 
``Logic Functions and Equations",
Springer, 2004.


 \bibitem{Sev21}  Severi, Francesco, ``Vorlesungen \"{u}ber Algebraische Geometrie", Teubner, Berlin, 1921.
 
\bibitem{SSY} Staton, Sam, 
Dario Stein, Hongseok Yang, Nathanael L.~Ackerman, Cameron E. Freer, and Daniel M. Roy, \emph{The Beta-Bernoulli process and algebraic effects},
{Proc.\ 45th.\ ICALP} (eds.\ Ioannis Chatzigiannakis, Christos Kaklamanis, D{\'{a}}niel Marx, and Donald Sannella), 
LIPIcs, {\bf 107}, 141:1--141:15, Schloss Dagstuhl, 2018.

\bibitem{Tay79} Taylor, Walter, \emph{Survey 79: Equational logic}, Houston Journal of Mathematics, \textbf{5}(S) (1979), 1--83.

\end{thebibliography}
\end{document}